\documentclass[
  reprint,
  superscriptaddress,
  amsmath,
  amssymb,
  aps,
  nofootinbib
]{revtex4-2}

\usepackage{graphicx}
\usepackage{dcolumn}
\usepackage{bm}
\usepackage{qcircuit}
\usepackage{physics}
\usepackage{mathtools}
\usepackage{xcolor}
\usepackage{booktabs}
\usepackage{multirow}
\usepackage[version=4]{mhchem}
\usepackage{amsthm}
\usepackage{hyperref}
\usepackage{cleveref}
\usepackage{braket}
\usepackage[ruled,vlined,linesnumbered]{algorithm2e}
\usepackage{float} 

\renewcommand{\cref}[1]{(\ref{#1})}
\def\lpar{\left(}
\def\rpar{\right)}

\def\Htar{H_\text{target}}

\newcommand{\bea}{\begin{eqnarray}}
\newcommand{\eea}{\end{eqnarray}}
\newcommand{\beaa}{\begin{eqnarray*}}
\newcommand{\eeaa}{\end{eqnarray*}}
\newcommand{\del}[3]{\left1 3 \right2}

\newtheorem{theorem}{Theorem}

\newcommand{\bi}{\begin{itemize}}
\newcommand{\ei}{\end{itemize}}
\newcommand{\bc}{\begin{center}}
\newcommand{\ec}{\end{center}}

\newcommand{\bal}[1]{
\begin{align} #1 \end{align}
}

\def\ntraj{N_\text{traj}}

\newcommand{\QEval}{Q_{\textit{eval}}}

\newcommand{\Ito}{It\^{o} }
\usepackage{xspace}

\newcommand{\qiskit}[1][]{\textsc{Qiskit}\ifx#1\empty\else~(v#1)\fi\xspace}
\newcommand{\qiskitnature}[1][]{\textsc{Qiskit Nature}\ifx#1\empty\else~(v#1)\fi\xspace}
\begin{document}
\title{Path Integral Quantum Control for Quantum Chemistry Applications
}

\author{Peyman Najafi}
\email{peyman.najafi@donders.ru.nl}
\affiliation{ Radboud University, Heyendaalseweg 135, 6525 AJ Nijmegen, The Netherlands}

\author{Aar\'on Villanueva}
\affiliation{ Radboud University, Heyendaalseweg 135, 6525 AJ Nijmegen, The Netherlands}

\author{Hilbert Kappen}
\affiliation{ Radboud University, Heyendaalseweg 135, 6525 AJ Nijmegen, The Netherlands}

\begin{abstract}
The Path integral Quantum Control (PiQC) algorithm was recently introduced in Ref.~\cite{villanueva2025stochastic} as a new approach for computing optimal controls in open and closed quantum systems.
Originally proposed for pulse-based quantum control, PiQC estimates optimal controls through global averages over quantum trajectories.
In this work, we adapt the PiQC algorithm to optimize parametrized quantum circuits by showing that the quantum circuit can be randomized using a continuous dynamics governed by a stochastic Schrödinger equation that is compatible with the path integral control framework.
In this adaptation, the circuit parameters become the controls to be optimized within PiQC.
We refer to this instance of PiQC as the Gate-based PiQC (GB-PiQC) algorithm.
We apply GB-PiQC for ground state preparation of electronic structure problems.
We benchmark the gate-based and pulse-based versions of PiQC against the Variational Quantum Eigensolver (VQE), which is optimized using the common Simultaneous Perturbation Stochastic Approximation (SPSA) optimizer, on a set of standard molecular Hamiltonians: \ce{H2}, \ce{LiH}, \ce{BeH2}, and \ce{H4}, mapped to 2-, 4-, 6-, and 6-qubit systems, respectively.
For each molecule, the benchmark is implemented at different bond distances, after performing a hyperparameter tuning of each algorithm at a fixed bond distance near the equilibrium geometry.
We find that both PiQC algorithms exhibit greater robustness than SPSA to variations in the target Hamiltonian induced by changes in molecular bond distances. 
Furthermore, PiQC algorithms also achieve superior performance compared to SPSA in most instances, particularly at stretched bond lengths, where the Hartree–Fock solution becomes less accurate and its error grows relative to equilibrium. 
\end{abstract}

\maketitle


\section{Introduction}

Quantum algorithms for quantum chemistry represent one of the most promising near-term applications for quantum devices \cite{tilly2022variational, grimsley2019adaptive}.
For ground state energy estimation of molecular Hamiltonians, hybrid quantum-classical approaches such as the Variational Quantum Eigensolver (VQE) \cite{peruzzo2014variational}, which belongs to the broader class of Variational Quantum Algorithms (VQAs), have been widely used.
Despite their success, VQAs face persistent challenges related to trainability, accuracy, and efficiency \cite{cerezo2021variational}.

Recent works have approached the ground state preparation problem from a different perspective, by optimizing control fields directly at the pulse level rather than parametrized gate sequences at the circuit level~\cite{magann2021pulses}.
Notable examples include Variational Quantum Optimal Control (VQOC) \cite{dekeijzer2023pulse} and ctrl-VQE \cite{meitei2021gatefree}, which treat control amplitudes as variational parameters to guide the system’s evolution.

The Path Integral Quantum Control (PiQC) algorithm~\cite{villanueva2025stochastic} was recently proposed for pulse-based quantum control of analog platforms in the presence or absence of dissipation.
The algorithm optimizes coherent control pulses by combining the stochastic Schrödinger equation (SSE) with path integral control techniques drawn from stochastic optimal control theory.
What differentiates PiQC from other standard quantum control approaches lies in the optimization process, which uses adaptive importance sampling over quantum trajectories to define an update rule scheme to iteratively improve the controls.
For the control of closed systems, PiQC implements an annealing schedule in the dissipation part that works as a parametrized ansatz to guide the optimization.
It has been noted in recent literature that noise can be used as a resource to enhance the performance of control protocols in variational quantum circuits~\cite{foldager2022noise, ito2023santaqlaus, liu2025stochastic} and analog platforms~\cite{levy2017action, guimaraes2024optimized, sveistrys2025speeding}.
In this regard, PiQC can be classified as a noise-assisted quantum control technique for the control of closed systems.

In this work, we adapt the PiQC algorithm to optimize parametrized quantum circuits, in concrete, VQEs.
We bridge the gap between pulse-based and gate-based quantum control by drawing on the observation that a randomized version of the VQE can be recast as a continuous stochastic dynamics driven by Wiener noise, which is suitable for a PiQC formulation.
This mapping enables the use of PiQC for quantum circuits, where the circuit parameters are mapped to control pulses in the stochastic dynamics.
Our work realizes the perspective outlined in~\cite{magann2021pulses}, which emphasized the potential of connecting quantum optimal control with circuit-level variational algorithms.

We benchmark the pulse-based and gate-based PiQC algorithms against a common VQE that uses Simultaneous Perturbation Stochastic Approximation (SPSA) for computing the ground state energy of molecular systems.
For this comparison, we tuned the hyperparameters of both methods through a limited search over small parameter sets, performed for each molecule at an interatomic distance near equilibrium (minimum energy).
We find that PiQC yields superior performance in most regimes.
For fixed tuned hyperparameters, PiQC is more robust across changes in the target problem (here, we varied the molecular interatomic distance), while SPSA is more sensitive to such variations.
Our results suggest that PiQC, both in pulse-based and gate-based form, offers a compelling alternative to variational methods for quantum chemistry applications.

The rest of the paper is organized as follows.
In Section~\ref{sec:piqc} we review the Path integral Quantum Control algorithm for pulse-based control and its application to closed quantum systems.
In Section~\ref{sec:GB_PiQC} we adapt the PiQC algorithm for VQEs, providing an exact mapping between the circuit action and a stochastic continuous dynamics compatible with the path integral control formulation.
After outlining our SPSA implementation in Section~\ref{sec:SPSA}, in Section~\ref{sec:Drift_Hamiltonian} we describe the drift Hamiltonian and hardware assumptions used in our simulations.
In Section~\ref{sec:faircomparison} we explain the conditions for fair comparison between the gate-based and pulse-based algorithms and specify the circuit ansatz used for simulations.
Finally, in Section~\ref{sec:results} we compare both gate-based and pulse-based PiQC algorithms with the VQE optimized via SPSA.
We conclude in Section~\ref{sec:conclusion} with an analysis of our findings.


\section{Background: the path integral quantum control algorithm}\label{sec:piqc}

\textit{Control problem definition.---}The Path integral Quantum Control (PiQC) algorithm~\cite{villanueva2025stochastic} uses the path integral control theory~\cite{kappen2005linear, thijssen2015path} to formulate a class of deterministic open-loop quantum control problems as a stochastic optimal control problem, and approximate the optimal controls as iterated averages over continuous quantum trajectories.

Given an objective function, PiQC computes coherent control pulses for a class of open quantum systems whose dynamics are governed by the Lindblad master equation with dissipators satisfying an anti-Hermitian property~\cite{villanueva2025stochastic}.
For closed quantum systems, PiQC computes optimal controls by gradually reducing environmental coupling strengths throughout the optimization process, annealing them down to small (but non-zero) values.
By the end of the optimization, the open-system dynamics closely approximates unitary evolution, with the resulting controls providing a good approximation to the dissipation-free ideal case.

In this work we focus on applications of the annealing variant of PiQC for ground state preparation of unitary dynamics, and refer the reader to~\cite{villanueva2025stochastic} for further details.
In PiQC, one begins by defining an open-loop control problem for an open system dynamics in Lindblad~\cite{breuer2002theory} form
\bal{\label{eq:lindblad}
\dot \rho = -i[H_0 + u_a H_a] \rho + D_{ab} \lpar H_a \rho H_b -\frac{1}{2} \{H_a H_b, \rho\} \rpar\,.
}
where $H_0,\, H_a\, (a=1, \ldots, n_c)$ are Hermitian operators, $u_a\, (a=1, \ldots, n_c)$ are the control fields, and $D$ is a positive definite symmetric matrix.
From now on, repeated indices imply summation.
In the limit of vanishing environmental coupling $D$, the dynamics above reduces to the dynamics of the closed system we aim to control.
We assume that the initial state $\rho(0)$ is pure.

Assume a total time interval $[0, T]$, with $T$ the final time, and define the cost objective
\bal{\label{eq:det-cost}
C[u] =  \frac{Q}{2}\Tr(\Htar \rho(T)) + \frac{1}{2} \int_0^T u(t)^\top R u(t) dt
}
with $\Htar$ the target Hamiltonian whose ground state we want to prepare, $\rho(T)$ is the state at the final time $T$, $Q$ is a positive real number, $R$ is a positive matrix, and $u$ the vector of controls.
The first term in~\cref{eq:det-cost} is the end cost and encodes the observable we aim to optimize.
The second term is commonly referred to as the fluence and penalizes the overall control energy.
Including this term is mandated by the path integral control framework.
The values $Q$ and $R$ are tunable parameters that regulate the trade-off between the end cost and the fluence.

We remark that the particular form of the open dynamics in~\cref{eq:lindblad} is not arbitrary, and is chosen that way in order to apply the path integral control formalism~\cite{villanueva2025stochastic}.
Given a suitable space of open-loop controls $u$, the control problem for closed systems consists in minimizing the objective~\cref{eq:det-cost} with dynamics~\cref{eq:lindblad} in the limit of vanishing environmental coupling $D$.


\textit{SOC formulation.---}The control problem with dynamics~\cref{eq:lindblad} and cost objective~\cref{eq:det-cost} is turned into a SOC problem by simulating the open dynamics with a stochastic Schrödinger equation (SSE) of the form
\bal{\label{eq:linear-sse}
d\psi(t) =& -i H_0 \psi(t) dt -i H_a \psi(t) (u_a(t) dt + dW_a(t)) \nonumber\\
&-\frac{1}{2} D_{ab} H_a H_b \psi(t) dt
 \,,
}
where $dW_a(t)$ are Wiener increments satisfying $\mathbb{E}(dW_a(t) dW_b(t)) = D_{ab} dt$.
The expectation $\mathbb{E}(\cdot)$ is taken over all noise realizations $W_{0:T}$ in the time interval $[0, T]$.
The last term in~\cref{eq:linear-sse} is such that the dynamics preserves the norm of the state, i.e. $d\|\psi(t)\|^2 = 0$ in the \Ito sense.
It is easy to show using \Ito calculus that~\cref{eq:linear-sse} is an unraveling of the Lindblad equation~\cref{eq:lindblad} in the sense that the evolution of the average state over all noise realizations $\rho(t) = \mathbb{E}(\psi(t) \psi(t)^\dagger)$ follows~\cref{eq:lindblad}.
We say that equation~\cref{eq:linear-sse} is in path integral form in the sense that it fits the type of dynamics that can be addressed within the path integral control theory~\footnote{%
We briefly recall that a SSE is in path integral form if it can be written as $d\psi_t = f(t, \psi_t) dt + g(t, \psi_t)(u(t) dt + dW(t))$ for some vector-valued complex functions $f$ and $g$. See~\cite{villanueva2025stochastic}.
}.

For open-loop controls, the cost objective~\cref{eq:det-cost} can be equivalently written as an average over quantum trajectories governed by~\cref{eq:linear-sse}, i.e.
\bal{\label{eq:cost-S}
C[u] = \mathbb{E}(S^u)
}
where
\bal{\label{eq:stoch-cost}
S^u =& \frac{Q}{2}\braket{\psi(T)| \Htar | \psi(T)} + \frac{1}{2} \int_0^T u(t)^\top R u(t)
\nonumber\\
&+ \frac{1}{2}\int_{0}^{T} u(t)^\top R dW(t)
}    
where $dW(t)$ is the vector of Wiener increments.
This equivalence can be easily seen by taking the expectation of $S^u$, noticing that the last integral vanishes under the expectation since it represents an \Ito integral, and identifying $\rho(T) = \mathbb{E}(\psi(T) \psi(T)^\dagger)$.
Then the deterministic control problem can be formulated as a SOC problem with cost objective~\cref{eq:cost-S} and stochastic dynamics~\cref{eq:linear-sse} which can be solved using path integral control techniques.

\textit{Adaptive importance sampling.---}We define the open-loop control model.
Consider a partition of the time interval $[0, T]$ such that $t_0 = 0 < t_1 < \ldots < t_K = T$, and define $u_a$ as a piece-wise constant control composed of $K$ pulses, i.e.
\bal{
u_a(t) = \sum_{k=1}^K u_{ak} \bm{1}_{I_k}(t) \quad a=1, \ldots, n_c
}
with $u_{ak}$ time-independent constants and $\bm{1}_{I_k}$ indicator functions such that $\bm{1}_{I_k}(t) = 1$ if $t \in I_k = [t_{k-1}, t_k]$ and zero otherwise, for $k=1, \dots, K$.

Provided that the matrices $R$ and $D$ are related by $R=D^{-1} \lambda$ for some real $\lambda > 0$ --this is a necessary requirement of the path integral formalism-- the PiQC algorithm gives an update rule to compute the optimal pulses using adaptive importance sampling (AIS)~\cite{villanueva2025stochastic, thijssen2015path}.
Starting with an initial guess for the control pulses $u^{(0)}_{ak}$, at the optimization step $p$ the update rule for $u^{(p)}$ is given by
\bal{\label{eq:importance_sampling_update}
u^{(p+1)}_{ak} = u^{(p)}_{ak} + \mathbb{E} \lpar \omega^{(p)} \frac{\Delta W_{ak}}{\Delta t_k} \rpar,
\quad k = 1, \dots, K\,,
}
where $\Delta W_{ak} := \int_{I_k} dW_a(t)$, $\Delta t_k = t_k - t_{k-1}$, and the expectation $\mathbb{E}(\cdot)$ is computed over trajectories sampled using the controls $u^{(p)}_{ak}$ computed in the previous step.
The importance sampling weights are given by
\bal{\label{eq:omega_p}
\omega^{(p)} := \frac{e^{-S^{(p)}/\lambda}}{\mathbb{E}\lpar e^{-S^{(p)}/\lambda} \rpar}\,,
}
where $S^{(p)}$ denotes the stochastic cost $S^u$ defined in~\cref{eq:stoch-cost} using the control $u = u^{(p)}$ computed at the previous step.

Intuitively, the control pulses $u_{ak}$ are iteratively improved using formula~\cref{eq:importance_sampling_update} in directions where the quantum trajectories, which are steered using the current control, contribute to minimize the cost objective.
This information is encoded in the importance sampling weights $\omega^{(p)}$.

\textit{PiQC for unitary dynamics.---}At a fixed noise matrix $D$, the AIS rule~\cref{eq:importance_sampling_update} gives a recipe to approximate the optimal controls that minimize the cost~\cref{eq:det-cost} for the open dynamics~\cref{eq:lindblad}.
In order to compute the optimal controls for the corresponding closed dynamics in the limit of $D \rightarrow 0$, one introduces an annealing schedule in the environmental coupling $D_j\, (j=0, \ldots, n_D-1)$, where $n_D$ is the number of annealing steps.
The schedule is defined such that it monotonically decreases the values of $D$ across optimization steps.
If by the end of the optimization $D$ is sufficiently small, the resulting controls $u_{ak}$ will approximate well the optimal controls of the corresponding unitary control problem.

For practical purposes, assume that the noise matrix $D$ is proportional to the identity, with $D$ representing now a positive real number.
Because the relation $R=D^{-1}\lambda$, $R$ also becomes a real number.
In Alg.~\cref{alg:annealedPiQC} we summarize the PiQC algorithm for ground state preparation with exponential noise schedule~\cite{villanueva2025stochastic}, which is used later in the numerical experiments.
The workflow of the algorithm is as follows.
At the start of the optimization, we set an initial guess for the control pulses $u^{(0)}_{ak}$ and define the annealing schedule $D_j\,(j=0, \ldots, n_D-1)$ with $D_\text{init}$ and $D_\text{final}$ being the extreme values.
For each value $D_i$ we run the AIS optimization given by formula~\cref{eq:importance_sampling_update} for a number $n_s$ of steps (steps 1-6 in Alg.~\cref{alg:annealedPiQC}).
Each AIS step consists of sampling $\ntraj$ trajectories using the SSE~\cref{eq:linear-sse} with the current control values, and collecting the statistics necessary for the control update.
The value $n_s$ of AIS steps must be such that the algorithm has enough time to optimize for that particular value of $D_j$.
This procedure is repeated iteratively until the annealing schedule reaches its final value $D_\text{final}$.
Upon completion of the optimization process (step 7 in Algorithm~\cref{alg:annealedPiQC}), we collect the energy expectation values computed from the final states at time $T$ across all quantum trajectories, and select the minimum value as our estimate of the target Hamiltonian's ground state energy.
The corresponding final state $\rho(T) \approx \ket{\text{GS}} \bra{\text{GS}}$ provides a high-fidelity approximation of the ground state $\ket{\text{GS}}$ of the target Hamiltonian.
\begin{algorithm}[H]
\SetAlgoLined
\SetInd{0.5em}{0.5em}
\SetKwData{Left}{left}\SetKwData{This}{this}\SetKwData{Up}{up}
\SetKwFunction{Union}{Union}\SetKwFunction{FindCompress}{FindCompress}
\SetKwInOut{Input}{input}\SetKwInOut{Output}{output}
\Input{Initial controls $\{u^{(0)}_{ak}\}$, initial noise $D_{\text{init}}$, final noise $D_{\text{final}}$, number of trajectories $\ntraj$, number of AIS steps $n_s$, number of annealing steps $n_D$}
\Output{$E$}
\BlankLine
\For{$j=0$ \KwTo $ n_D-1$}{
$D_j \gets D_{\text{init}} \times \left( \frac{D_{\text{final}}}{D_{\text{init}}} \right)^{j / (n_D - 1)}$\\
    \For{$p=0$ \KwTo $n_s-1$}{
    Update controls using noise matrix $D_j$ with $\ntraj$ trajectories with using~\cref{eq:importance_sampling_update}\\
    }
}
From the final iteration, collect energies $E^{(i)}  = \braket{\psi(T)^{(i)} |H | \psi(T)^{(i)}}$ for $i = 1, \dots, \ntraj$\\
\Return $E = \min_i E^{(i)}$
\\
\caption{PiQC with Exponential Noise Schedule}\label{alg:annealedPiQC}
\end{algorithm}


\section{Gate-Based Path Integral Quantum Control} \label{sec:GB_PiQC}

In this Section we adapt the PiQC algorithm for variational quantum eigensolvers (VQEs).
VQEs were first introduced in~\cite{peruzzo2014variational} as hybrid quantum-classical algorithms designed to estimate the ground state energy of a quantum system.
Given a parametrized circuit ansatz $U(\bm{\theta})$ with $\bm{\theta} \in \mathbb{R}^L$, a VQE prepares a quantum state $\ket{\psi(\bm{\theta})}$ that approximates the ground state of the target Hamiltonian $\Htar$ by
iteratively updating the circuit parameters to minimize the expected energy $E(\bm{\theta}) = \braket{\psi(\bm{\theta}) | \Htar | \psi(\bm{\theta})}$.

Consider a general parameterized quantum circuit of the form
\begin{equation}
\label{eq:general_ansatz}
U(\boldsymbol{\theta}) 
= \prod_{\ell=1}^{L} \Big[ V^{(\ell)}\, R^{(\ell)}(\boldsymbol{\theta}_\ell) \Big],
\end{equation}
composed of $L$ unitary blocks where in each block $\ell$, $V^{(\ell)}$ represents a generic unitary operator independent of the parameters $\boldsymbol{\theta}$ of the form $e^{-i H_0^{(\ell)}}$ with $H_0^{(\ell)}$ a local Hermitian operator, and the local rotations $R^{(\ell)}(\boldsymbol{\theta}_\ell)$ are defined by
\begin{equation}
R^{(\ell)}(\boldsymbol{\theta}_\ell) = \prod_{m=1}^{M} \prod_{q=1}^{n} R^{(\ell)}_{q, m}(\theta^{(\ell)}_{q,m}),
\label{eq:R_loc_l}
\end{equation}
where the $R^{(\ell)}_{q, m}$ are single-qubit rotations of the form $R_{q, m}^{(\ell)}(\theta^{(\ell)}_{q,m}) = e^{-i \theta^{(\ell)}_{q,m} H^{(\ell)}_{q,m}/2}$ with $H^{(\ell)}_{q,m} \in \{\sigma^x_q, \sigma^y_q, \sigma^z_q\}$ Pauli operators acting on qubit $q$. 
The products in $\ell, m$ are taken in reversed order to align with the circuit convention, meaning that gates with smaller indices are applied earlier.
In Fig.~\ref{fig:hardware_efficient_ansatz_with_time} we illustrate the block architecture of the ansatz given in Eq.~\cref{eq:general_ansatz}.
\begin{figure*}
    \centering
\includegraphics[width=0.6\linewidth]{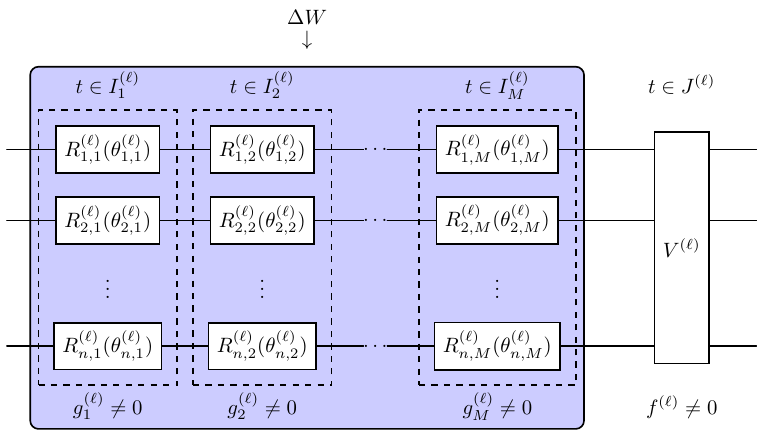}
\caption{Block $\ell$ corresponding to the general ansatz defined in~\ref{eq:general_ansatz}.
The ansatz consists of $L$ repetitions of this pattern.}  
\label{fig:hardware_efficient_ansatz_with_time}
\end{figure*}
Quantum circuits of the form~\cref{eq:general_ansatz} and~\cref{eq:R_loc_l} provide a natural choice for VQAs: their gradients can be evaluated efficiently for most objective functions~\cite{mcclean2018barren}, and their flexible structure accommodates widely used ansatz families, such as the hardware-efficient ansatz (HEA) for quantum chemistry applications~\cite{kandala2017hardware}.

To apply PiQC to the parametrized circuit~\cref{eq:general_ansatz}, we first introduce Wiener noise into the circuit parameters, allowing the action of $U(\bm{\theta})$, given an initial state $\ket{\psi_0}$, to be reformulated as a stochastic Schrödinger equation compatible with the path integral control formalism~\cite{villanueva2025stochastic}.

We assume, without loss of generality, that each rotation $R^{(\ell)}_{q, m}$ and drift $V^{(\ell)}$ layers act for a total time of length one.
Then, each block $\ell$ acts in a time interval of length $M+1$, and the entire circuit acts for a time interval $[0, T]$ with $T=L(M+1)$.
In the time interval $[0, T]$ consider the following SSE
\begin{widetext}
\begin{equation}
    d\psi(t)  = - i \sum_\ell f^{(\ell)}(t)H^{(\ell)}_0 \psi(t) dt 
    - \sum_{q, m, \ell} \Big[ \frac{i}{2} \left(\theta^{(\ell)}_{q,m} dt + dW_q(t) \right) g^{(\ell)}_m(t) H^{(\ell)}_{q,m} \psi(t) + \frac{1}{8} g^{(\ell)}_m(t)D\psi(t) dt \Big],
    \label{eq:continous_dynamic}
\end{equation}
\end{widetext}
where the functions $f^{(\ell)}(t)$ and $g^{(\ell)}_m(t)$ are defined as
\begin{align}
& f^{(\ell)}(t) =
\begin{cases}
1, & \text{if } t \in  J^{(\ell)}, \\
0, & \text{otherwise}
\end{cases}\\
    & g^{(\ell)}_m(t) =
\begin{cases}
1, & \text{if } t \in  I^{(\ell)}_m, \\
0, & \text{otherwise},
\end{cases}
\end{align}
and disjoint intervals $J^{(\ell)}$ and $I^{(\ell)}_m$ defined as
\begin{align}
    I^{(\ell)}_m =& \big[ (\ell - 1)(M+1) + m - 1, \nonumber \\
    &(\ell - 1)(M+1) + m \big]
    \label{eq:Ilcm}\,, \\
    J^{(\ell)} =& \big[(\ell - 1)(M+1) + M, \nonumber \\
    &(\ell - 1)(M+1) + M + 1 \big]\,.
    \label{eq:Il0}
\end{align}
The symbols $\theta^{(\ell)}_{q,m}$ represent time-independent constants and the Wiener increments $dW_q(t)$ satisfy $\mathbb{E}(dW_q(t) dW_{q'}(t)) = \delta_{q,q'} D dt$ with $D$ a positive constant.
\begin{theorem}[Randomized circuit as continuous-time dynamics]
\label{th:stoch-circuit-map}
The state $\psi(T)$ generated by the stochastic dynamics given by~\cref{eq:continous_dynamic} from the initial time $t=0$ to the final time $T=L(M + 1)$ is equivalent to the state generated by the parameterized circuit given by~\cref{eq:general_ansatz} with gate parameters and entangling blocks given by 
\begin{align}
    & \Tilde{\theta}^{(\ell)}_{q,m}  = \theta^{(\ell)}_{q,m}  +  \Delta W^{(\ell)}_{q,m} 
    \label{eq:u_to_theta_1}\\
    &V^{(\ell)}  = e^{-i H^{(\ell)}_0} 
\end{align}
where $\Delta W^{(\ell)}_{q,m} = \int_{I^{(\ell)}_m} dW_q(t)$, and $I^{(\ell)}_{m}$ is given by~\cref{eq:Ilcm}.
\end{theorem}

\begin{proof}
    Based on the definition of $f^{(\ell)}(t)$ and $g^{(\ell)}_m(t)$ at each $t \in [0, T]$ only one of these functions is nonzero, which depends on which interval includes $t$.  
    For $t \in J^{(\ell)}$ in which only $f^{(\ell)}(t)$ is nonzero, the dynamic is a unitary dynamic within a unit time interval evolving under the Hamiltonian $H^{(\ell)}_0$. Therefore, the propagator for this time interval is $e^{- i H^{(\ell)}_0}$, which is equal to $V^{(\ell)}$. 
    For $t \in I^{(\ell)}_m$ in which $g^{(\ell)}_m(t)$ is nonzero, define an effective $H_{eff}(t)$: 
    \begin{equation}
        H_{eff}(t) = \frac{1}{2}\int_{t^{(\ell)}_m}^{t} \sum_q \left( \theta^{(\ell)}_{q,m} dt + dW_q(t) \right) H^{(\ell)}_{q,m},
    \end{equation}
    where $t^{(\ell)}_m$ is the initial time in the interval $I^{(\ell)}_m$.
    Then $\psi(t) = e^{-i H_{eff}\left(t - t^{(\ell)}_m\right)}\psi(t^{(\ell)}_{m})$ is the solution of continuous dynamics in~\cref{eq:continous_dynamic}.
    This follows directly by taking the differential using \Ito calculus: 
    \begin{align*}
         d\psi(t) = - \sum_q \big[ &\frac{i}{2} \left( \theta^{(\ell)}_{q,m} dt + dW_q(t) \right) H^{(\ell)}_{q, m} \psi(t) \nonumber \\ &+ \frac{1}{8} D \psi(t) dt \big].
    \end{align*}
    As a result the propagator for this time interval is $\prod_{q=1}^{n} e^{-i \Tilde{\theta}^{(\ell)}_{q,m} H^{(\ell)}_{q,m}/2}$, which makes the propagator for time interval $I^{(\ell)} := \bigcup_{m=1}^M I^{(\ell)}_m$ to be equal to $R^{(\ell)}(\boldsymbol{\Tilde{\theta}}_l)$.   Therefore, by knowing the propagator for each of these disjoint intervals ($J^{(\ell)}$ and $I^{(\ell)}$), and the sequential order in which  $f^{(\ell)}(t)$ and $g^{(\ell)}_m(t)$ are defined, it becomes clear that the evolved state at final time $T$, matches the ansatz defined in~\cref{eq:general_ansatz}. 

\end{proof}

Following~\cite{villanueva2025stochastic}, we can derive an AIS update rule for the control pulses $\bm{\theta} = \{\theta^{(\ell)}_{q,m}\}$  appearing in~\cref{eq:u_to_theta_1}, that is similar to the AIS formula given by~\cref{eq:importance_sampling_update}.
Write the stochastic cost~\cref{eq:stoch-cost} in terms of the circuit parameters $\bm{\theta}$ as
\bal{\label{stoch-cost-circuit}
S^{\bm{\theta}} = \frac{Q}{2}\braket{\psi(\bm{\theta})| \Htar | \psi(\bm{\theta})} + \frac{R}{2} \bm{\theta}^\top \bm{\theta} + \frac{R}{2} \bm{\theta}^\top \Delta \bm{W}\,.
}
Then, AIS update rule is given by
\bal{\label{eq:IS-rule-circuit}
\bm{\theta}^{(p+1)} = \bm{\theta}^{(p)} + \mathbb{E}(\omega^{(p)} \Delta \bm{W})\,,
}
where $\Delta \bm{W} = \{\Delta W^{(\ell)}_{q,m} \}$ is the vector of integrated Wiener noises, see Theorem~\cref{th:stoch-circuit-map}.
The importance sampling weight $\omega^{(p)} = \frac{e^{-S^{(p)}/\lambda}}{\mathbb{E}(e^{-S^{(p)}/\lambda})}$ is defined in the same way as in~\cref{eq:importance_sampling_update}, with the stochastic cost now written as in~\cref{stoch-cost-circuit}.
The optimization problem consists in minimizing $C[\bm{\theta}] = \mathbb{E}(S^{\bm{\theta}})$ in the limit of $D\rightarrow  0$.
The proof of~\cref{eq:IS-rule-circuit} is straightforward~\cite{villanueva2025stochastic} and, for completeness, we write it in Appendix~\ref{app:IS-rule-circuit}.

We remark an important point.
In practice, in order to update the gate parameters of a circuit of the form~\cref{eq:general_ansatz} with PiQC, we do not generate continuous trajectories using equation~\cref{eq:continous_dynamic}.
Instead, we start with a random initialization of the circuit parameters $\bm{\theta}$ and generate $N_{\text{traj}}$ random states $\ket{\psi(\bm{\theta} + \Delta\bm{W})}$ using the quantum circuit.
Each state corresponds to an independent realization of the Wiener increment $\Delta \bm{W}$.
Each integrated Wiener increment $\Delta W^{(\ell)}_{q,m} = \int_{I^{(\ell)}_m} dW_q(t)$ is an independent Gaussian random variable of mean zero and variance $D |I^{(\ell)}_m|= D$, since $I^{(\ell)}_m$ has length one.
In the first stages of the optimization with annealing PiQC, since $D$ takes a large value, the relatively high variance of random circuit trajectories enables more effective exploration of the control landscape.
This leads to global statistics that dictate the direction of circuit parameter updates $\bm{\theta}$ via the AIS rule~\cref{eq:IS-rule-circuit}.
Note that the computational cost of computing the expected energy $E(\bm{\theta})$ at each optimization step does not scale with the number of circuit parameters, but it is tied to the number of trajectories that are necessary to achieve statistical accuracy.
This is in contrast to other standard methods for optimizing circuits, such as parameter-shift~\cite{mitarai2018quantum, schuld2019evaluating}.


\section{Simultaneous Perturbation Stochastic Approximation (SPSA)}\label{sec:SPSA}

In order to compare our results with other optimizers, we use the well-known simultaneous perturbation stochastic approximation (SPSA) algorithm~\cite{spall1987stochastic}, an optimization method that is well-suited for high-dimensional problems where obtaining gradient information is prohibitive.
SPSA estimates the gradient of the cost objective using only two evaluations of the cost function per iteration, which makes the computational cost per step independent of the number of circuit parameters.
Given a parameter vector \( \boldsymbol{\theta} \in \mathbb{R}^L \) and an the energy cost $E(\bm{\theta})$, at optimization step $j$ the update rule is given by~\cite{spall1998implementation}:
\begin{equation}
\boldsymbol{\theta}^{(j+1)} = \boldsymbol{\theta}^{(j)} - a_j \, \hat{\bm{g}}_j,
\end{equation}

where \(a_j\) is the learning rate and \( \hat{g}_j \) is the stochastic gradient estimate:
\begin{equation}
\hat{\bm{g}}_j = \frac{E(\boldsymbol{\theta}^{(j)} + c_j \boldsymbol{\Delta}_j) - E(\boldsymbol{\theta}^{(j)} - c_j \boldsymbol{\Delta}_j)}{2 c_j} \, \boldsymbol{\Delta}_j^{-1},
\end{equation}

where \( \boldsymbol{\Delta}_j \in \mathbb{R}^L \) (the inverse is taken element-wise) is a random perturbation vector drawn from a product of independent Rademacher distributions --i.e., each entry is sampled independently from a symmetric Bernoulli distribution over \(\{+1, -1\}\) with equal probability 0.5, and \(c_j\) controls the perturbation magnitude. In our numerical experiments, we use SPSA with a fixed learning rate $a_j = a$ and a fixed perturbation magnitude $c_j = c$ for the reasons explained in App.~\ref{app:spsa_hyperparams}.

\section{Drift Hamiltonian and Hardware Assumptions}\label{sec:Drift_Hamiltonian}

Rydberg atom platforms combine long coherence times, the flexibility of movable qubits, strong long-range interactions, and the ability to coherently switch interaction strengths, making them particularly well-suited for applications in quantum information science and quantum simulation \cite{saffman2010quantum, ludmir2024modeling, weber2017calculation}. Here, following~\cite{dekeijzer2023pulse}, we model the drift Hamiltonian based on a Rydberg-atom quantum computing platform with neutral atoms arranged in a one-dimensional array. Qubits are encoded in a ground--Rydberg state pair, and the system features always-on van der Waals interactions between atoms in the Rydberg state. The drift Hamiltonian is given by~\cite{dekeijzer2023pulse}:
\begin{equation}
H_0 = \sum_{i<j} \frac{C_6}{r_{ij}^6} \ket{11}_{ij}\bra{11}, 
\label{eq:drift_hamiltonian}
\end{equation}

where \( r_{ij} = r |i - j| \) is the distance between atoms \( i \) and \( j \), \( C_6 \) is the van der Waals interaction strength coefficient, and \( \ket{11}_{ij} \bra{11} \) is the projector onto both atoms being in the Rydberg state. 
The atoms are assumed to be equally spaced along a line, with nearest-neighbor distance \( r \). This configuration gives rise to the long-range interactions decaying with the sixth power of the distance.

In neutral-atom systems driven by nearly monochromatic laser fields which are in resonance with the energy difference between the Rydberg and the ground states (no detuning), the interacting Hamiltonian term on a single qubit can be written as \cite{morgado2021quantum, dekeijzer2023pulse}:

\[
\frac{\Omega(t)}{2} \left( e^{i\phi(t)} \ket{0}\!\bra{1} + e^{-i\phi(t)} \ket{1}\!\bra{0} \right)\,.
\]
The parameter $\Omega(t)$ is the coupling strength and $\phi(t)$ denotes the phase of the laser acting on the qubit at time $t$. Therefore the control Hamiltonian acting on all the qubits can be modeled as 
\begin{equation}
H_c(t) = \sum_{q=1}^n \left( u_{qx}(t)\, \sigma^x_q + u_{qy}(t)\, \sigma^y_q \right),    
\end{equation}
where \( u_{qx}(t) \) and \( u_{qy}(t) \) are real-valued control amplitudes for the \( x \)- and \( y \)-axis rotations on qubit \( q \).


\section{Benchmarking methodology}
\label{sec:faircomparison}
In order to have a fair comparison between the gate-based and pulse-based algorithms, we need the system to evolve under similar circumstances (as much as possible).
The conditions for this fair comparison between a gate-based and pulse-based algorithm are discussed in~\cite{dekeijzer2023pulse}, and here we follow the same analysis.
For both gate-based and pulse-based ansatzes, the system is driven by the same drift Hamiltonian.
The only difference lies in the time-dependent interaction terms: pulse controls for the pulse-based approach and gate parameters for the gate-based approach.
Following~\cite{dekeijzer2023pulse}, we have chosen a HEA with the generic unitary operator $V^{(\ell)} = e^{-i\tau_V H_0}$, where $H_0$ is given in Eq.~(\ref{eq:drift_hamiltonian}), and the characteristic entanglement time is defined as \( \tau_V = 1 / V \), where \( V = C_6 / r^6 \) denotes the Van der Waals interaction strength between the two nearest neighbor qubits at distance $r$.
In each layer of the ansatz, we have three consecutive local rotations ($M=3$) generated by $H^{(\ell)}_{q,m} = \sigma^z_q, \sigma^x_q, \sigma^z_q$ for $m=1, 2, 3$, respectively, and each qubit $q$.

Timescales are used to express control durations and rescale simulation parameters consistently throughout the work.
We adopt a Rabi frequency of \( \Omega_R = 1\,\text{kHz} \), which sets the single-qubit gate time to \( \tau_g = 1/\Omega_R = 1\,\text{ms} \)~\cite{dekeijzer2023pulse}.
The interaction strength is tuned to \( V = 0.1\,\text{kHz} \), by adjusting the interatomic spacing \( r \) to a suitable value, so that $\tau_V=10$ms. These values satisfy the regime \( \tau_V \gg \tau_g \)~\cite{dekeijzer2023pulse}.
The total circuit duration for the gate-based algorithm is then modeled as \( T_{\text{GB}} = L(\tau_g + \tau_V) \), with \( L \) denoting the ansatz depth. Therefore for having the equivalent evolution processes for both algorithms, we need to fix the time of the pulse-based algorithm $T_{PB}$ to be equal to \( L(\tau_g + \tau_V) \). 

Apart from having similar evolution we also need to consider the total number of Hamiltonian evaluations given by the expected energy $E(\bm{\theta})$, which we will refer to as $\QEval$.
For a fair comparison, the algorithms should be compared under a fixed $\QEval$ budget.

In the case of SPSA, requires two quantum evaluations per iteration.
Thus, $\QEval = 2 \times I_{\text{SPSA}}$ where $I_{\text{SPSA}}$ denotes the number of iterations of the SPSA optimizer.

Both PiQC algorithms, GB-PiQC and PB-PiQC, require $\ntraj$ quantum evaluations per optimization step. Thus, $\QEval = \ntraj \times I_{\text{PiQC}}$.
In our benchmarks, the total quantum evaluations for the three algorithms are set to be equal.


\section{Results}
\label{sec:results}
\begin{figure*}[t]
    \centering
    \includegraphics[width=.9\linewidth]{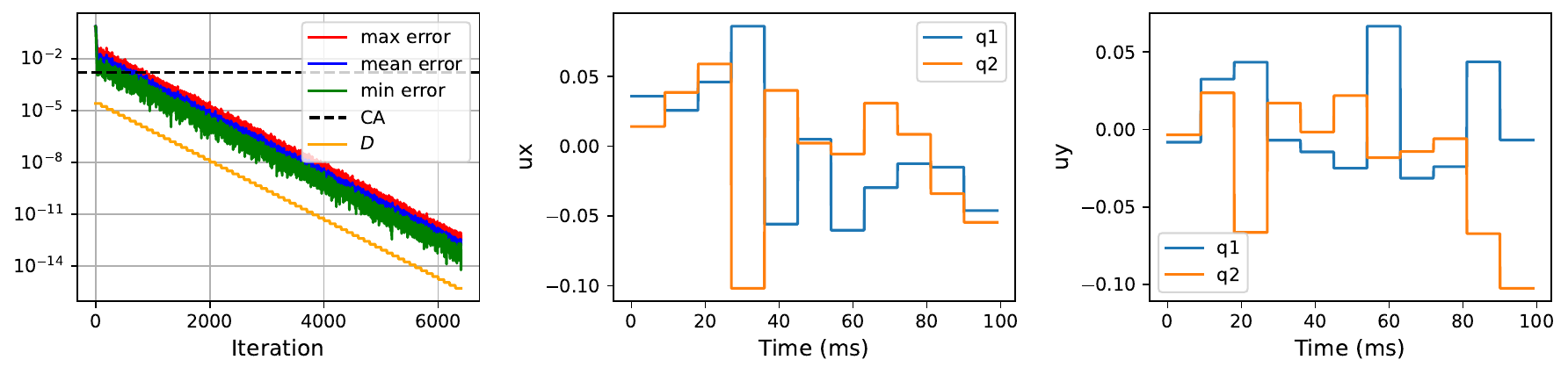}
    \vspace{0.75em} 
    \includegraphics[width=.9\linewidth]{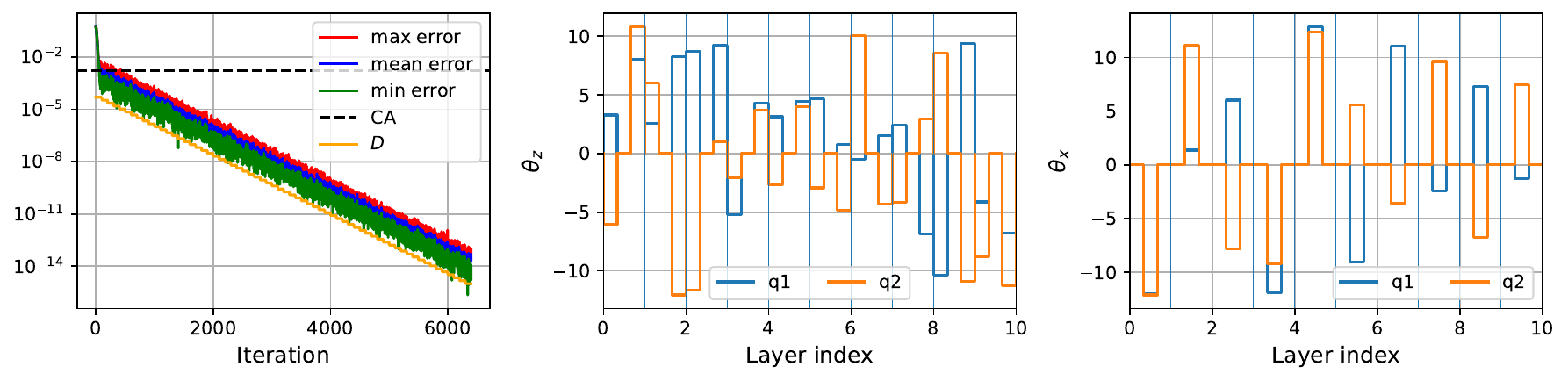}
    \caption{Application of the annealed PiQC algorithm to estimate the ground state of \ce{H2} at a bond distance of 0.79~\AA{}. 
\textbf{Top row (PB-PiQC):} \emph{Left}--evolution of the energy error (minimum, mean, and maximum across 10 trajectories) together with the annealing schedule for $D$, which gradually reduces the effective noise level. \emph{Middle}--optimized piecewise-constant $X$ control as a function of time for each qubit. \emph{Right}--optimized piecewise-constant $Y$ control for each qubit.
\textbf{Bottom row (GB-PiQC):} \emph{Left}—the same diagnostics for the gate-based run. \emph{Middle}—rotational angles of $R_Z$ gate plotted for different layer indexes, where each layer occupies a unit interval; within each layer, the first third shows the $\theta_z$ corresponding to the first $R_Z$ rotation gate and the last third shows the one corresponding to the last $R_Z$ rotation gate in each layer. \emph{Right}—rotational angles of $R_X$ gate $\theta_x$, occupying the middle third of each layer, corresponding to the middle $R_X$ rotation gate in each layer. Both PiQC algorithms lead to very low errors as the iterations progress.}
    \label{fig:H2_single_distance}
\end{figure*}
\begin{figure*}[ht]
    \centering
    \includegraphics[width=.9\linewidth]{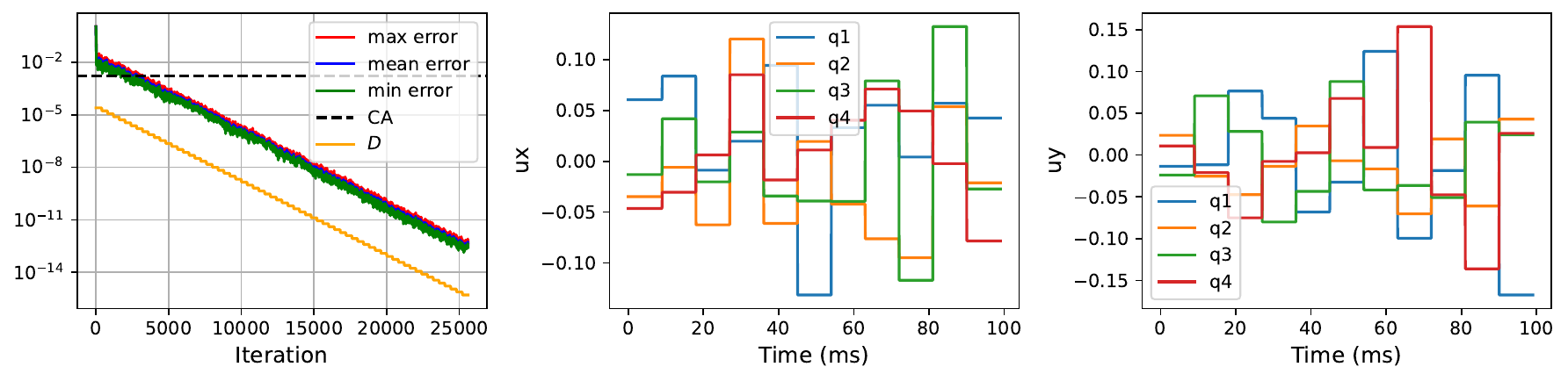}
    \vspace{0.75em} 
    \includegraphics[width=.9\linewidth]{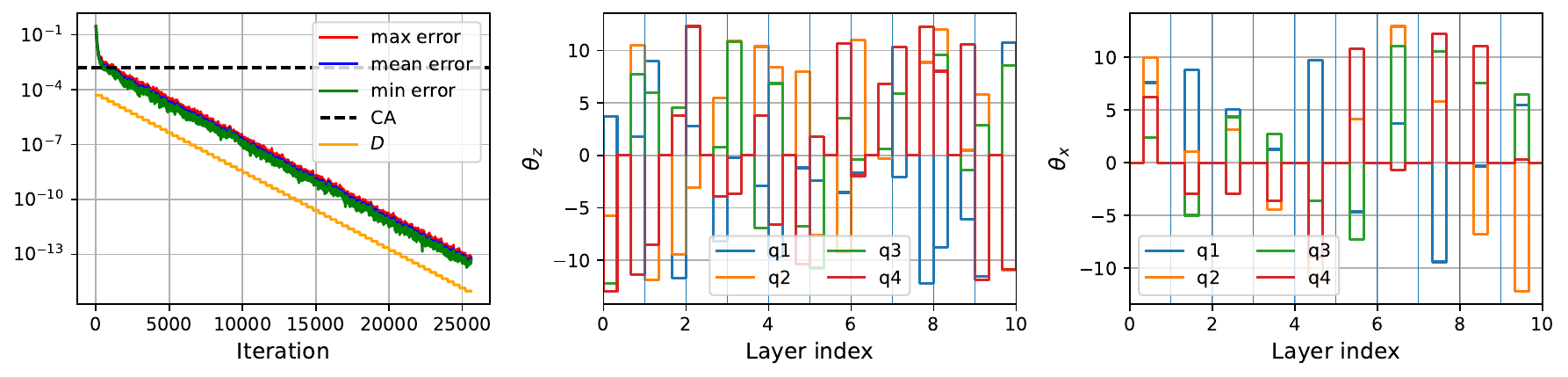}
\caption{Same plotting conventions as Fig.~\ref{fig:H2_single_distance}.
Results for \ce{LiH} at bond distance of 1.33~\AA{}. 
Top row (PB--PiQC): energy-error trace with $D$ schedule (left), gate parameters (middle/right).
Bottom row (GB--PiQC): energy-error trace with $D$ schedule (left), gate parameters (middle/right).}
    \label{fig:LiH_s1.33e_distance}
\end{figure*}
We prepare the ground state of  electronic structure problems for a set of small molecules: \ce{H2}, \ce{LiH}, \ce{BeH2}, and \ce{H4}. 
Molecular Hamiltonians were generated using \texttt{PySCFDriver} implemented in \qiskitnature package~\cite{the_qiskit_nature_developers_and_contrib_2023_7828768} with the STO-3G basis set, and subsequently mapped to qubit operators using various fermion-to-qubit transformations and reduction techniques from \qiskitnature tailored to each system.

For \ce{H2}, following Ref.~\cite{dekeijzer2023pulse}, the active space is restricted to the $1s$ orbitals, with higher-energy orbitals neglected. The Hamiltonian is mapped to 4 qubits using the parity transformation and subsequently reduced to 2 qubits via the standard two-qubit reduction ($\mathbb{Z}_2$ symmetry reduction) in \qiskitnature.

In the case of \ce{LiH}, the \texttt{ActiveSpaceTransformer} from \qiskitnature is used to select two molecular orbitals located near the Fermi level. The resulting Hamiltonian is then mapped to qubits using the Jordan--Wigner transformation, yielding a 4-qubit Hamiltonian.

For \ce{BeH2}, the molecule is aligned linearly with the \ce{Be} atom at the center and the two \ce{H} atoms placed symmetrically along the x-axis. The core molecular orbital is frozen, and the \ce{Be} $2p_y$ and $2p_z$ orbitals (molecular orbitals 3 and 4) are removed using the \texttt{FreezeCoreTransformer} in \qiskitnature to define a reduced active space. The parity mapping with $\mathbb{Z}_2$ symmetry reduction is then applied, yielding a 6-qubit Hamiltonian.

Finally, for \ce{H4}, similar to~\cite{dekeijzer2023pulse}, the hydrogen atoms are positioned linearly with uniform spacing. All four $1s$ orbitals are treated as active. The Hamiltonian is mapped using the parity mapping with $\mathbb{Z}_2$ symmetry reduction, resulting in a 6-qubit Hamiltonian.

The chosen molecules exhibit both weakly and strongly correlated ground states across different interatomic distances, making them suitable benchmarks for assessing quantum algorithms in varied regimes.
Exact ground state energies were obtained via direct diagonalization of the mapped qubit Hamiltonian as a reference, and energy errors were calculated with respect to this reference value.

To ensure fair and stable comparisons across molecular systems, we performed light hyperparameter tuning for both the SPSA and PiQC algorithms near the equilibrium geometry (lowest energy) of each molecule.
Specifically, hyperparameter tuning was carried out on the following interatomic distances: \ce{H2} at 0.72~\AA{}, \ce{LiH} at 1.60~\AA{}, \ce{BeH2} at 1.339~\AA{}, and \ce{H4} at 0.54~\AA{}.

For SPSA, we tested a small set of logarithmically spaced learning rates together with perturbation magnitudes selected as fixed fractions of the learning rate.
In the case of PiQC, the step size of the annealing schedule on the diffusion parameter \(D\) was explored, while other hyperparameters such as the number of trajectories, and initial and final values of \(D\), were fixed based on heuristic choices and preliminary experiments.
The tuning was deliberately kept limited in order to reduce computational cost and is intended to provide reasonable, rather than fully optimized, performance.
Since SPSA involves stochastic perturbations and PiQC incorporates Wiener increments, the expectation value outcomes vary with the random seed.
Therefore, for each hyperparameter configuration, we repeated the optimization using 5 different seeds and considered the median final error as the representative value. See App.~\ref{app:spsa_hyperparams} and App.~\ref{app:piqc_hyperparams} for details. 

In this work, we performed all the simulations related to quantum circuits with \qiskit[1.3.1]~\cite{qiskit2024}, and evaluated observables with \qiskit’s Estimator without any hardware or sampling noise. 


\subsection{Example: PiQC for Electronic Structure Problems}

We illustrate the application of PiQC for two representative electronic structure problems. 
Specifically, we consider the ground-state energy estimation of the \ce{H2} and \ce{LiH} molecules at interatomic distances of 0.79~\AA{} and 1.33~\AA{}, respectively. These examples serve to demonstrate the effectiveness of PiQC in optimizing the control landscapes for finding the ground state of molecular Hamiltonians.
Note that in all of our numerical experiments with PiQC, we have set the control parameter $R=1$ and the end cost parameter $Q$ to large values $Q\gg 1$ in order to achieve good end cost ($E$) results. 

In Fig.~\ref{fig:H2_single_distance} we show the performance of PiQC for \ce{H2} molecule.
Noise schedule, alongside the energy error and final calculated optimal controls are plotted.
The horizon time for the pulse-based case is $T=99$ ms, which corresponds to $L=9$ for the gate-based case, and piecewise constant controls were used with $K=11$ for pulse-based PiQC.
Following the notation of Algorithm~\ref{alg:annealedPiQC} for the noise schedule, $D_{init} = 2.5 \times 10^{-5}$, $D_{final} = 5 \times 10^{-16}$, $n_s=100$ and $n_D=64$. The number of trajectories $\ntraj$ used here and in all other experiments is 10, and $\QEval=6.4\times 10^{4}$. 
We define the energy error as the difference between the variational energy and the exact ground state energy obtained via exact diagonalization of the mapped qubit Hamiltonian. 

Fig.~\ref{fig:LiH_s1.33e_distance} shows a similar plot for \ce{LiH} molecule with same $D_{init}$ and $D_{final}$ as above, stepsize $n_s$ is 400 iterations, $n_D=64$ and $\QEval=2.56\times 10^{5}$. For both molecules, the error decreases exponentially with iterations. 
\begin{figure}[!ht]
    \centering
    \includegraphics[width=.8\linewidth]{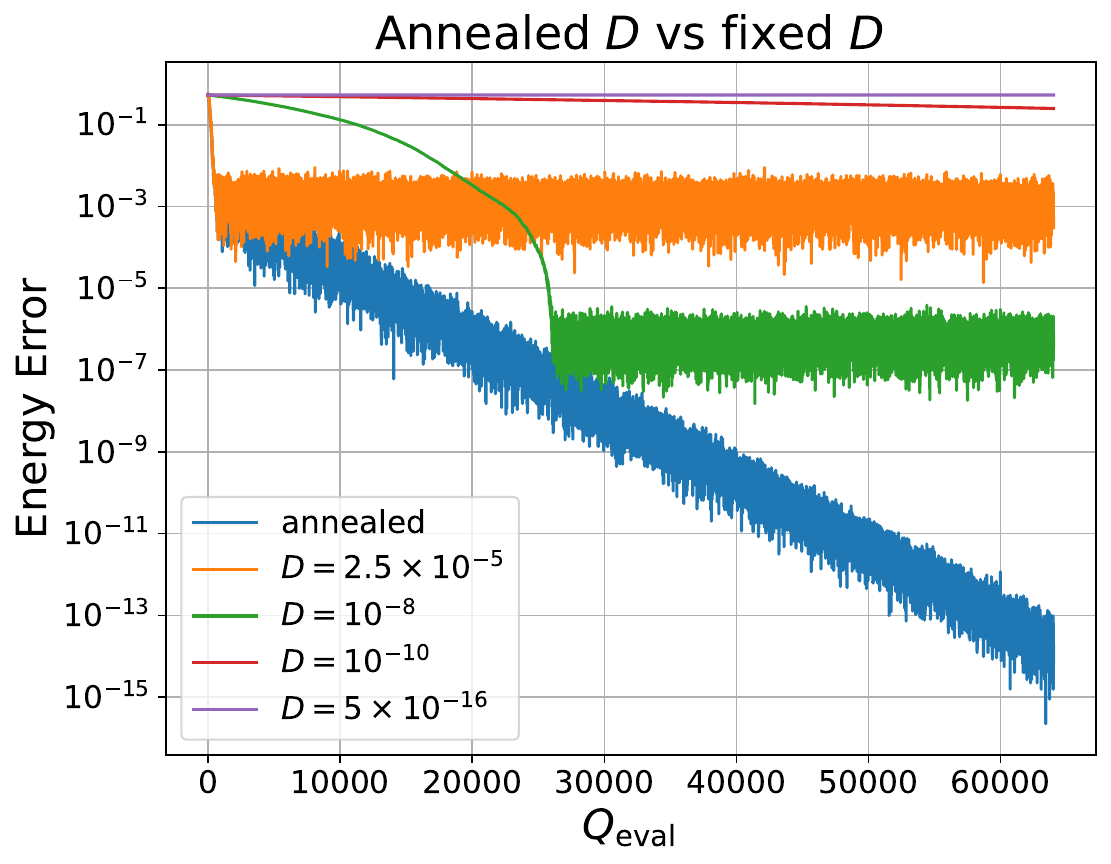}
    \caption{Energy error comparison between annealed and fixed $D$ parameters using GB-PiQC for \ce{H2} ground state estimation.
    In the annealed case, the diffusion parameter $D$ starts at $D_{\text{init}} = 2.5 \times 10^{-5}$ and exponentially decreases to $D_{\text{final}} = 5 \times 10^{-16}$.
    For the fixed-$D$ cases we use four different values of $D$ in decreasing order: $D_{\text{init}}\,, 10^{-8},\, 10^{-10},\, 5\times  10^{-16}$. 
    For fixed $D$ values, larger parameters ($D = D_{\text{init}}$) show rapid initial error reduction followed by saturation and fluctuations due to high noise.
    Smaller fixed values ($D = 10^{-8}, 10^{-10}, 5 \times 10^{-16}$) exhibit progressively slower convergence with diminishing error changes.
    We obtain analogous results using PB-PiQC.
    }
\label{fig:annealing_advantage}
\end{figure}

In Fig.~\ref{fig:annealing_advantage} we demonstrate that gradually reducing the parameter $D$ to small $D_\text{final}$ outperforms using small but fixed $D$, achieving better convergence within fewer steps.
For this, we apply GB-PiQC to compute the energy error across optimization steps using four different but fixed values of $D$, compared to a single run with annealing.
Beyond a certain number of optimization steps, the accuracy for fixed $D$ tends to saturate, while in the annealing case the accuracy shows exponential improvement over the same optimization window, achieving several orders of magnitude lower errors than the fixed $D$ cases.
This behavior arises from the sampling efficiency inherent in the quantum trajectory simulation process defined by~\cref{eq:IS-rule-circuit}.
The parameter $D$ controls the magnitude of stochastic fluctuations in the quantum trajectories through the Wiener increments $\Delta W^{(\ell)}_{q,m}$, which are Gaussian variables of variance $D$.
When $D$ is too small from the start, each quantum trajectory carries similar information, requiring a large number of trajectories $\ntraj$ to obtain reliable parameter updates and significantly slowing down convergence.
This is evident in the case where $D=5 \times 10^{-16}$ in Fig.~\cref{fig:annealing_advantage}.
Conversely, starting with larger values of $D$ introduces greater diversity among quantum trajectories, enabling more efficient exploration of the control landscape.
This diverse sampling allows reliable parameter estimation with fewer trajectories, accelerating the optimization process.
\begin{figure*}[!t]
    \centering
    \includegraphics[width=\linewidth]{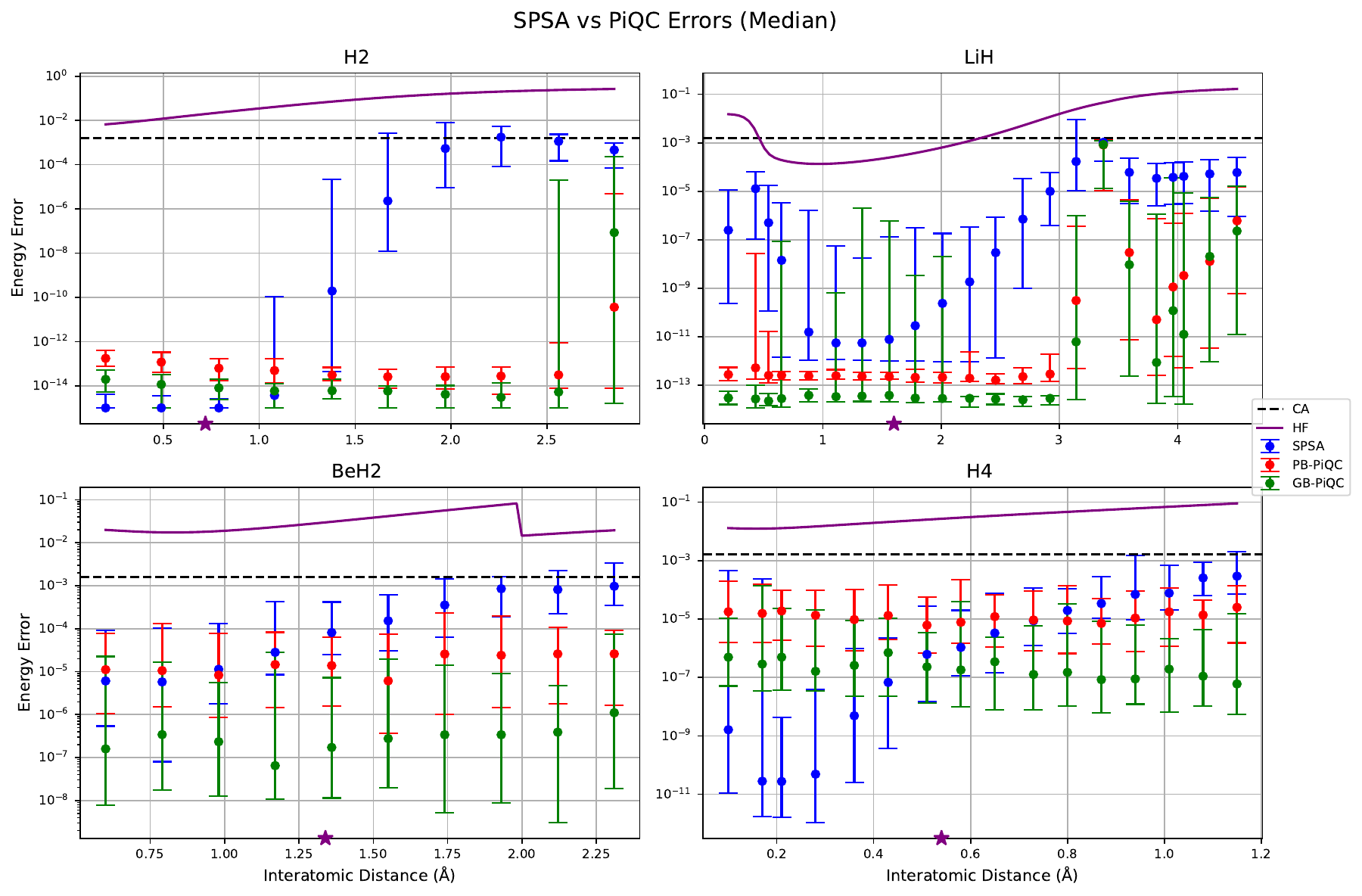}
    \caption{
Comparison of SPSA and PiQC across varying bond distances for four molecules: \textbf{\ce{H2}}, \textbf{\ce{LiH}}, \textbf{\ce{BeH2}}, and \textbf{\ce{H4}} (top left to bottom right). Each plot reports results across 20 different seeds, showing the median error (circle marker) and the range between minimum and maximum errors (vertical cap lines). The purple asterisk on the x-axis marks the bond distance used for hyperparameter tuning for each molecule. For \textbf{\ce{H2}}, both PB-PiQC and GB-PiQC exhibits more consistent performance than SPSA. SPSA's error increases with bond distance, and its worst-case error exceeds chemical accuracy at some large bond distances. In contrast for both PiQC algorithms, the worst-case error remains below this threshold throughout the distances. For \textbf{\ce{LiH}}, both PiQC algorithms outperforms SPSA across all bond distances in terms of median error. SPSA’s worst-case error crosses the chemical accuracy line at one bond distance, while both PiQC’s worst-case errors stays entirely below it. In \textbf{\ce{BeH2}}, PiQC algorithms again shows superior consistency. SPSA’s median error rises with bond distance, and its worst-case error exceeds chemical accuracy at the two largest bond distances, while for PiQC algorithms, the worst-case errors remain consistently below the threshold. Finally, in \textbf{\ce{H4}}, SPSA performs better than PiQC algorithms at short bond distances, but its performance deteriorates with increasing bond distance. Its worst-case error eventually surpasses chemical accuracy, whereas for PiQC algorithms the worst-case errors stay below chemical accuracy at all bond distances.
}
    \label{fig:benchmarks}
\end{figure*}

\subsection{PiQC and SPSA comparison}

Here we benchmark SPSA against PB-PiQC and GB-PiQC by first tuning the hyperparameters for each molecule near the equilibrium interatomic distance, then using the algorithm with the tuned hyperparameters over a range of distances. Since both PiQC algorithm and SPSA optimizer are stochastic (different realizations of the Wiener noise $W(t)$ and $\Delta$ lead to different outcomes) and the final result also depends on the initial guess of controls and gate parameters~\footnote{For PB-PiQC, initial controls were fixed to zero, and for SPSA and GB-PiQC, the gate parameters and virtual controls were sampled uniformly from the $[-2\pi, 2\pi]$.}, 20 different seeds were used. Median error over different seeds is reported. 

Fig.~\ref{fig:benchmarks} shows the benchmark results, where 
median error (circle marker) and the range between minimum and maximum errors (vertical cap lines) are plotted. $\QEval$ was set to $6.4\times10^{4}$ for \ce{H2}, and to $2.56\times10^{5}$ for \ce{LiH}, 
\ce{BeH2} and \ce{H4}. 

For \ce{H2} molecule, by increasing the bond distance, the median error of SPSA rises significantly, indicating sensitivity to changes in the target Hamiltonian. In contrast, both PiQC algorithms maintains a relatively stable and low error across the same range. This suggests that PiQC is more robust to variations in the underlying molecular Hamiltonian and generalizes better across problem instances.

For \ce{LiH} molecule, both PiQC algorithms outperforms SPSA at all the distances and similar to \ce{H2} case, PiQC is less sensitive towards changes of the bond distance. 

For \ce{BeH2} and \ce{H4} molecules, similar to previous cases, PiQC remains largely unaffected by variations in the target Hamiltonian across the bond distance ranges considered, while the final error of SPSA gradually increases by increasing the bond distance and comes close to exceeding the chemical accuracy threshold. In \ce{BeH2} case, GB-PiQC outperforms SPSA in all the distances, and PB-PiQC outperforms SPSA in all except the two smallest distances. In \ce{H4} case, SPSA outperforms PiQC algorithms at small distances, but after distance gets larger than approximately 0.45~\AA{}, GB-PiQC has the lowest median error. 

Overall, PiQC algorithms shows superior consistency over varying bond distances, indicating its robustness towards target Hamiltonian changes. Furthermore, for SPSA algorithm there is at least one distance for each molecule in which the maximum error of SPSA becomes larger than the chemical accuracy, while both PiQC's maximum errors are always below the chemical accuracy.


\section{Conclusion}\label{sec:conclusion}

PiQC was recently proposed as a novel generic method for pulse-based control of closed and open systems~\cite{villanueva2025stochastic}.
In this work, we successfully adapted the PiQC algorithm for optimizing variational quantum eigensolvers, developing a gate-based (GB-PiQC) variant that shows compelling advantages over common approaches such as SPSA.
Like SPSA, our approach uses a fixed number of objective function evaluations per optimization step, making it particularly suitable for high-dimensional parameterized quantum circuits.


We benchmarked both pulse-based (PB-PiQC) and gate-based (GB-PiQC) variants of the PiQC algorithm against SPSA for ground state preparation of electronic structure problems.
Performance evaluation was conducted across multiple benchmark molecules including \ce{LiH}, \ce{BeH2}, \ce{H2}, \ce{H4}.

For the \ce{LiH} molecule, GB-PiQC and PB-PiQC consistently outperformed SPSA across all bond distances tested.
For \ce{BeH2}, GB-PiQC achieved the best overall performance at all bond distances, while SPSA matched PB-PiQC's performance only at the shortest distances, with significant degradation at large bond distances.

For \ce{H2} and \ce{H4}, SPSA demonstrated superior performance at short bond distances but exhibited diminishing accuracy and reduced robustness as the bond distance increased.
Conversely, both PiQC algorithms consistently outperformed SPSA across the intermediate-to-large bond distance regime while maintaining stable accuracy throughout most of this range.

Our results indicate that when both PiQC and SPSA hyperparameters are tuned at a single bond length, PiQC shows more robustness than SPSA (as implemented in this study) to variations in the target Hamiltonian, and frequently achieves significantly higher accuracy than SPSA.
In terms of worst-case performance, both PiQC algorithms always reached chemical accuracy, whereas SPSA failed to achieve chemical accuracy in at least one instance across all four molecules.

For all molecules considered here, the Hartree-Fock error generally increases with bond distance, indicating that PiQC's superior performance over SPSA is most pronounced in regimes where mean-field approximations fail significantly.

Our results suggest that PiQC represents a promising approach for optimizing variational quantum algorithms, demonstrating both enhanced accuracy and robustness compared to conventional optimization methods.
However, comprehensive evaluation of PiQC's resilience to hardware imperfections and sampling noise remains essential for practical quantum computing implementations.


\begin{acknowledgments}
P.N. acknowledge support from the ‘Quantum Inspire—the Dutch
Quantum Computer in the Cloud’ project (NWA.1292.19.194) of the NWA research program
‘Research on Routes by Consortia (ORC)’, which is funded by the Netherlands Organization
for Scientific Research (NWO). This work used the Dutch national e-infrastructure with the support of the SURF Cooperative using grant no. EINF-12076.
\end{acknowledgments}


\section{Appendix}
\subsection{Proof of formula~\cref{eq:IS-rule-circuit}}\label{app:IS-rule-circuit}

For a general open-loop control model $u_a$ in the time interval $[0, T]$ such that
\bal{
u_a(t) = \sum_{k=1}^K A_{ak} h_k(t) \quad a=1, \ldots, n_c
}
where $A_{ak}$ are constants, and $\{h_k\}$ is the a set of $K$ time-dependent basis functions, it is shown in~\cite{villanueva2025stochastic} that the importance sampling update rule takes the form
\bal{
A^{(p+1)} = A^{(p)} + C^{(p)} B^{-1}
}
where $A^{(p+1)},\, A^{(p)}$ are matrices with components $A^{(p+1)}_{ak}, A^{(p)}_{ak}$, respectively, and $B$ and $C^{(p)}$ are matrices whose components are given by
\bal{
B_{kk'}=& \int_0^T h_k(t)h_{k'}(t) dt \label{eq:B}\\
C^{(p)}_{ak} =& \mathbb{E}\lpar \omega^{(p)} \int_0^T h_k (t) dW_a(t) \rpar \label{eq:C}
}
with $\omega^{(p)}$ as in~\cref{eq:omega_p}.
Note that $B$ is a fixed matrix that is computed once at the start of the optimization.

To prove formula~\cref{eq:IS-rule-circuit} we must specify the control model for the dynamics given by the SSE~\cref{eq:continous_dynamic}.
We define the control model as
\bal{
    u_q(t) = \sum_{m, \ell} \theta^{(\ell)}_{q,m} g^{(\ell)}_m(t) \quad q = 1, \ldots, n
}
where $\theta^{(\ell)}_{q,m}$ are the same as in~\cref{eq:continous_dynamic} with set of basis functions corresponding to $\{ g^{(\ell)}_m \}$.
By replacing this definition in formulas~\cref{eq:B}, ~\cref{eq:C} and taking into account that each time sub-interval has length one, we obtain formula~\cref{eq:IS-rule-circuit}.

\subsection{SPSA Hyperparameters}
\label{app:spsa_hyperparams}
The key hyperparameters in SPSA are the learning rate \( a \) and the perturbation size \( c \). We selected candidate values for \( a \) on a logarithmic scale: \( a \in \{10^{-1}, 10^{-2}, 10^{-3}, \ldots\} \). For each value of \( a \), the candidate values for \( c \) were chosen to be smaller by a factor of 5, 10, or 20, i.e., \( c \in \{a/5, a/10, a/20\} \). This grid of \((a, c)\) pairs was evaluated using multiple random seeds and the median value was chosen in order to account for the stochastic nature of the optimizer. The best combination was selected based on the lowest final energy after a fixed number of $\QEval$. 

Table~\ref{tab:spsa_hyperparams} summarizes the selected \((a, c)\) values for each molecule used in our experiments.

\begin{table}[h]
\centering
\caption{Selected SPSA hyperparameters for each molecule.}
\label{tab:spsa_hyperparams}
\begin{tabular}{lcc}
\toprule
\textbf{Molecule} & \textbf{Learning Rate \(a\)} & \textbf{Perturbation \(c\)} \\
\midrule
\(\ce{H2}\)     & 0.001 & .00005 \\
\(\ce{LiH}\)     & 0.01 & 0.0005 \\
\(\ce{BeH2}\)   & 0.001 & 0.00005 \\
\(\ce{H4}\)   & 0.001 & 0.00005 \\
\bottomrule
\end{tabular}
\end{table}

Instead of using a fixed perturbation and learning rate, it is common to employ decaying 
sequences of the form \cite{spall1998overview, kandala2017hardware}
\begin{align*}
    a_k &= \frac{a}{(A + k + 1)^{\alpha}}\\
    c_k &= \frac{c}{(k + 1)^{\gamma}},
    \label{eq:spsa_schedule}
\end{align*}
where $a, A, c, \alpha,$ and $\gamma$ are tunable hyperparameters. An initial calibration step (which often requires 50 cost function evaluations) is typically performed before optimization begins \cite{belaloui2025groundstate, kandala2017hardware}. In \qiskit, the default values are $c=0.2$, $\alpha=0.602$, $\gamma=0.101$, and $A=0$. 

The reason that we chose to proceed with fixed tuned learning rate and perturbation ($a, c$), instead of the default version of decaying learning rate and perturbation ($a_k, c_k$) in \qiskit with calibration, is that with these fixed tuned values we got lower errors and better performance as shown in Fig.~\ref{fig:SPSA_hyperparameter_effect}. The fixed learning rate and perturbation lead to better results than the decaying one for the $H_{\text{target}}$ cases considered here, and therefore we chose this better tuned version for a fair comparison with the PiQC.  

\begin{figure}[htbp]
    \centering
    \begin{minipage}{0.48\textwidth}
        \centering
        \includegraphics[width=\linewidth]{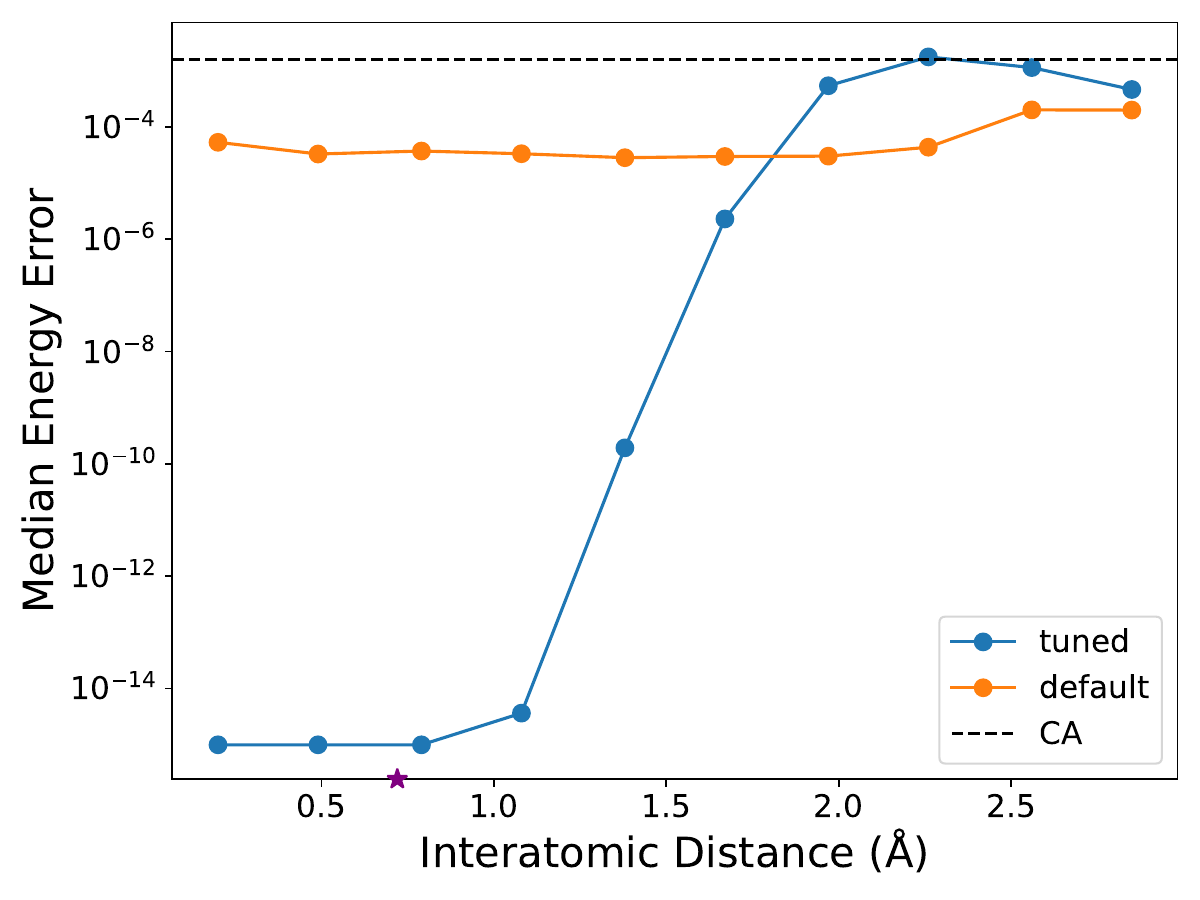}
    \end{minipage}
    \hfill
    \begin{minipage}{0.48\textwidth}
        \centering
        \includegraphics[width=\linewidth]{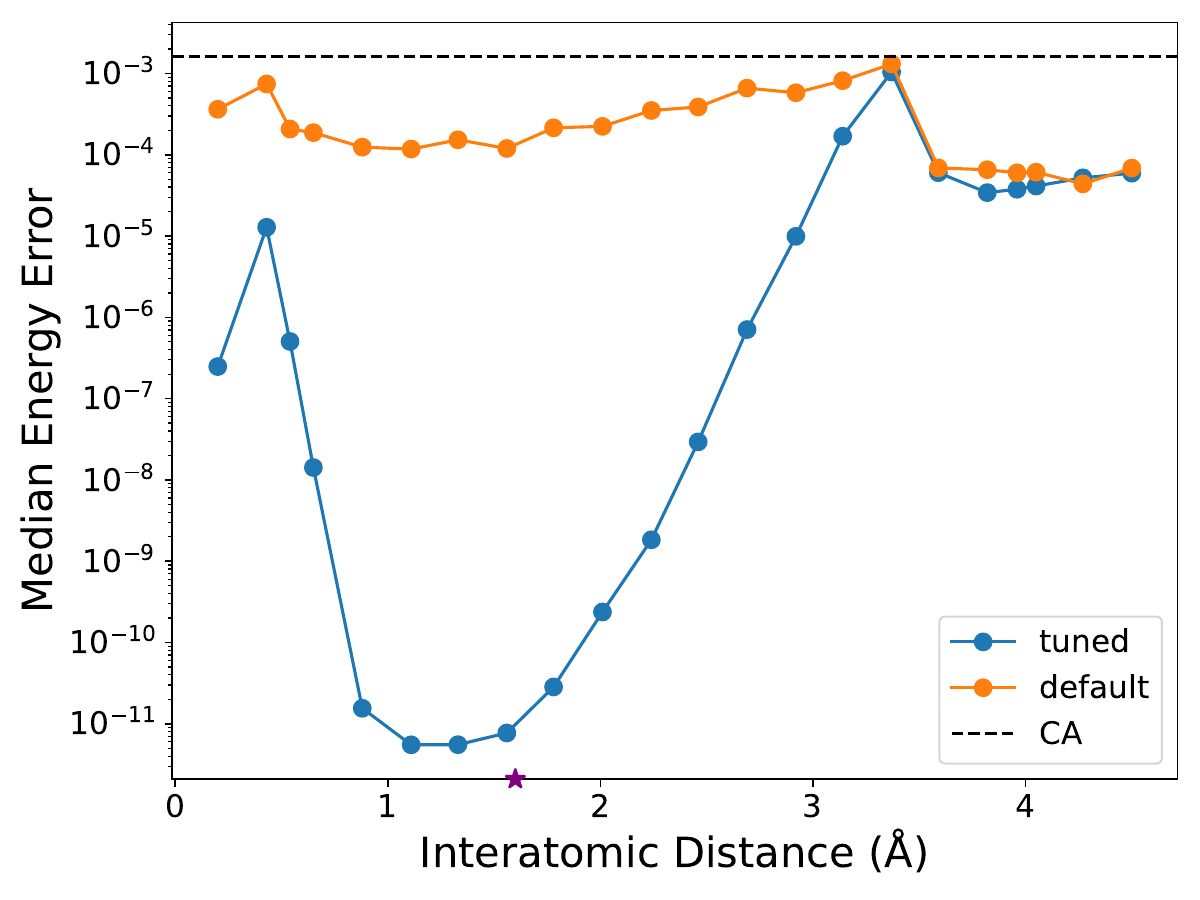}
    \end{minipage}
    \caption{Comparison of SPSA optimizer with different hyperparameter selections. Fixed learning rate and perturbation ($a, c$) after tuning with grid search method is indicated as ``tuned" and ``default" refers to the decaying learning rate and perturbation hyperparameters ($a_k, c_k$) that is used in \qiskit with a calibration step before optimizing. Top figure is for \ce{H2} and the bottom one belongs to \ce{LiH}. The purple asterisk on the x-axis marks the bond distance used for hyperparameter tuning of $(a, c)$ for the fixed learning rate an perturbation case. This figure shows that using fixed $(a, c)$ after tuning around equilibrium bond distance, overall lead to better performance than using the decaying $(a_k, c_k)$ learning rate and perturbation for the molecules considered.} 
    \label{fig:SPSA_hyperparameter_effect}
\end{figure}


\subsection{PiQC Hyperparameters}
\label{app:piqc_hyperparams}
In the PiQC experiments, as mentioned in Algorithm~\ref{alg:annealedPiQC}, we used an exponentially decreasing schedule for the noise parameter \(D\) from an initial value \(D_{\text{init}}\) to a final value \(D_{\text{final}}\). For PiQC we only tuned it on the PB-PiQC, and the final tuned schedule was good enough to be used for GB-PiQC as well, therefore a separate tuning on GB-PiQC was not done. 

The main hyperparameter that was searched over was the number of AIS steps $n_s$ at each $D$ value. Candidate $n_s$ values were \{50, 100, 200, 400, 800, 1600\}, and for each chosen value, the number of steps $n_D$ was determined as \(I_{\text{PiQC}} / n_s\), where \(I_{\text{PiQC}}\) denotes the total number of iterations. 

The initial value of \(D_{\text{init}} = 2.5 \times 10^{-5}\) was fixed across all molecules, and this was not tuned. The final value \(D_{\text{final}}\) was selected heuristically based on earlier experiments; for the smaller molecules (\ce{H2} and \ce{LiH}) it was set to \(D_{\text{final}} = 5 \times 10^{-16}\) and for the bigger ones (\ce{H2} and \ce{Lih}) \(D_{\text{final}} = 5 \times 10^{-13}\) was used. Each optimization run used 10 stochastic quantum trajectories ($\ntraj=10$), and the control parameters were initialized to zero.

Table~\ref{tab:PiQC_hyperparams} lists the selected PiQC noise schedule hyperparameters for each molecule.

\begin{table}[h]
\centering
\caption{Selected PiQC hyperparameters per molecule. \(D_{\text{init}} = 2.5 \times 10^{-5}\) is fixed.}
\label{tab:PiQC_hyperparams}
\begin{tabular}{lccc}
\toprule
\textbf{Molecule} & \textbf{Step Size $n_s$} & \textbf{Steps $n_D$}  & \(\boldsymbol{D_{\text{final}}}\) \\
\midrule
$\ce{H2}$     & 100 & 64 & \(5*10^{-16}\) \\
$\ce{LiH}$    & 400 & 64 & \(5*10^{-16}\) \\
$\ce{BeH2}$   & 800 & 32 & \(5*10^{-13}\) \\
$\ce{H4}$     & 800 & 32 & \(5*10^{-13}\) \\
\bottomrule
\end{tabular}
\end{table}

\bibliographystyle{quantum}
\bibliography{reference}

\end{document}